\documentclass[a4paper]{article}

  \usepackage{microtype}
  \usepackage[ruled,linesnumbered,noline,noend]{algorithm2e}
  \usepackage{amsmath,amsthm}
  \usepackage{fullpage}
  \usepackage{authblk}
  \usepackage{enumerate}
  \usepackage{hyperref}
  \usepackage[T1]{fontenc}
  \usepackage[utf8]{inputenc}
  \usepackage{lmodern}

\newtheorem{theorem}{Theorem}[section]
\newtheorem{fact}[theorem]{Fact}
\newtheorem{lemma}[theorem]{Lemma}
\newtheorem{observation}[theorem]{Observation}
\newtheorem{proposition}[theorem]{Proposition}
\newtheorem{corollary}[theorem]{Corollary}
\newtheorem{claim}[theorem]{Claim}
\newtheorem{definition}[theorem]{Definition}

\newcommand{\Oh}{\mathcal{O}}

\usepackage{mathtools}
\DeclarePairedDelimiter{\floor}{\lfloor}{\rfloor}

\newcommand{\Occ}{\mathsf{Occ}}
\newcommand{\Covered}{\mathsf{Covered}}
\newcommand{\StartOcc}{\mathsf{StartOcc}}
\newcommand{\PREF}{\mathsf{PREF}}
\newcommand{\lcp}{\mathsf{lcp}}
\newcommand{\rot}{\mathsf{rot}}

\renewcommand{\L}{\mathcal{L}}
\newcommand{\D}{\mathcal{D}}
\newcommand{\num}{\mathit{num}}
\newcommand{\suma}{\mathit{sum}}

\newcommand{\res}{\mathit{res}}
\newcommand{\Pred}{\mathit{Pred}}

\newcommand{\RM}{\mathit{RM}}

\newcommand{\Ham}{\mathit{Ham}}
\newcommand{\Lev}{\mathit{Lev}}
\newcommand{\ed}{\mathit{ed}}

 \newcommand{\defproblem}[3]{
  \vspace{2mm}
\noindent\fbox{
  \begin{minipage}{0.96\textwidth}
  \textsc{#1}\\
  {\bf{Input:}} #2  \\
  {\bf{Output:}} #3
  \end{minipage}
  }
  \vspace{2mm}
}

\begin{document}

\title{$k$-Approximate Quasiperiodicity under~Hamming~and~Edit~Distance}

\author[1]{Aleksander K{\k{e}}dzierski}
\author[1]{Jakub Radoszewski}

  \affil[1]{\normalsize Institute of Informatics,  University of Warsaw, Warsaw, Poland and \protect\\
    Samsung R\&D Institute, Warsaw, Poland, \protect\\
    \texttt{aa.kedzierski2@uw.edu.pl}, \texttt{jrad@mimuw.edu.pl}}

\date{\vspace{-.8cm}}

\maketitle

\begin{abstract}
Quasiperiodicity in strings was introduced almost 30 years ago as an extension of string periodicity.
The basic notions of quasiperiodicity are cover and seed.
A cover of a text $T$ is a string whose occurrences in $T$ cover all positions of $T$.
A seed of text $T$ is a cover of a superstring of $T$.
In various applications exact quasiperiodicity is still not sufficient due to the presence of errors.
We consider approximate notions of quasiperiodicity, for which we allow approximate occurrences in $T$
with a small Hamming, Levenshtein or weighted edit distance.

In previous work Sip et al.\ (2002) and Christodoulakis et al.\ (2005) showed that computing approximate covers and seeds, respectively,
under weighted edit distance is NP-hard.
They, therefore, considered restricted approximate covers and seeds which need to be factors of the original string $T$ and presented polynomial-time algorithms
for computing them.
Further algorithms, considering approximate occurrences with Hamming distance bounded by $k$, were given in several contributions by Guth et al.
They also studied relaxed approximate quasiperiods that do not need to cover all positions of $T$.

In case of large data the exponents in polynomial time complexity play a crucial role.
We present more efficient algorithms for computing restricted approximate covers and seeds.
In particular, we improve upon the complexities of many of the aforementioned algorithms, also for relaxed quasiperiods.
Our solutions are especially efficient if the number (or total cost) of allowed errors is bounded.
We also show NP-hardness of computing non-restricted approximate covers and seeds under Hamming distance.

Approximate covers were studied in three recent contributions at CPM over the last three years.
However, these works consider a different definition of an approximate cover of $T$, that is, the shortest exact cover of a string $T'$ with the smallest Hamming distance from $T$.
\end{abstract}

\section{Introduction}
Quasiperiodicity was introduced as an extension of periodicity~\cite{DBLP:journals/tcs/ApostolicoE93}. Its aim is to capture repetitive structure of strings that do not have an exact period.
The basic notions of quasiperiodicity are cover (also called quasiperiod) and seed.
A \emph{cover} of a string $T$ is a string $C$ whose occurrences cover all positions of $T$.
A \emph{seed} of string $T$ is a cover of a superstring of $T$.
Covers and seeds were first considered in \cite{DBLP:journals/ipl/ApostolicoFI91} and \cite{DBLP:journals/algorithmica/IliopoulosMP96}, respectively,
and linear-time algorithms computing them are known;
see \cite{DBLP:journals/ipl/Breslauer92,DBLP:journals/algorithmica/IliopoulosMP96,DBLP:journals/algorithmica/LiS02,DBLP:journals/ipl/MooreS94,DBLP:journals/ipl/MooreS95}
and \cite{10.1145/3386369}.

A cover is necessarily a \emph{border}, that is, a prefix and a suffix of the string.
A seed $C$ of $T$ covers all positions of $T$ by its occurrences or by left- or right-overhangs, that is, by suffixes of $C$ being prefixes of $T$
and prefixes of $C$ being suffixes of $T$.
In order to avoid extreme cases one usually assumes that covers $C$ of $T$ need to satisfy $|C|<|T|$ and seeds $C$ need to satisfy $2|C| \le |T|$
(so a seed needs to be a factor of $T$).
Seeds, unlike covers, preserve an important property of periods that if $T$ has a period or a seed, then every (sufficiently long) factor of $T$ has the same period or seed, respectively.

The classic notions of quasiperiodicity may not capture repetitive structure of strings in practical settings; it was also confirmed
by a recent experimental study~\cite{DBLP:journals/corr/abs-1909-11336}.
In order to tackle this problem, further types of quasiperiodicity were studied that require that only a certain number of positions in a string are covered.
This way notions of enhanced cover, partial cover and partial seed were introduced.
A \emph{partial cover} and \emph{partial seed} are required to cover a given number of positions of a string,
where for the partial seed overhangs are allowed, and an \emph{enhanced cover} is a partial cover with an additional requirement of being a border of the string.
$\Oh(n \log n)$-time algorithms for computing shortest partial covers and seeds were shown in \cite{DBLP:journals/algorithmica/KociumakaPRRW15}
and \cite{DBLP:journals/tcs/KociumakaPRRW16}, respectively, whereas a linear-time algorithm for computing a proper enhanced cover
that covers the maximum number of positions in $T$ was presented (among other variations of the problem) in \cite{DBLP:journals/tcs/FlouriIKPPST13}.

Further study has lead to \emph{approximate quasiperiodicity} in which approximate occurrences of a quasiperiod are allowed.
In particular, Hamming, Levenshtein and weighted edit distance were considered.
A \emph{$k$-approximate cover} of string $T$ is a string $C$ whose approximate occurrences with distance at most $k$ cover $T$.
Similarly one can define a \emph{$k$-approximate seed}, allowing overhangs.
These notions were introduced by Sip et al.~\cite{SimParkKimLee} and Christodoulakis et al.~\cite{DBLP:journals/jalc/ChristodoulakisIPS05}, respectively,
who showed that the problem of checking if a string $T$ has a
$k$-approximate cover and $k$-approximate seed, respectively, for a given $k$ is NP-complete under weighted edit distance.
(Their proof used arbitrary integer weights and a constant-sized---12 letters in the case of approximate seeds---alphabet.)
Therefore, they considered a \emph{restricted} version of the problem in which the approximate cover or seed is required to be a factor of~$T$.
Formally, the problem is to compute, for every factor of $T$, the smallest $k$ for which it is a $k$-approximate cover or seed of $T$.
For this version of the problem, they presented an $\Oh(n^3)$-time algorithm for the Hamming distance and an $\Oh(n^4)$-time algorithm for the edit distance\footnote{%
In fact, they consider \emph{relative} Hamming and Levenshtein distances
which are inversely proportional to the length of the candidate factor and seek for an approximate cover/seed that minimizes such distance.
However, their algorithms actually compute the minimum distance $k$ for every factor of $T$ under the standard distance definitions.}.
The same problems under Hamming distance were considered by
Guth et al.~\cite{DBLP:conf/itat/GuthMB08} and Guth and Melichar~\cite{Guth2010}.
They studied a \emph{$k$-restricted} version of the problems, in which we are only interested in factors of $T$ being $\ell$-approximate covers or seeds for $\ell \le k$,
and developed $\Oh(n^3(|\Sigma|+k))$-time and $\Oh(n^3|\Sigma|k)$-time automata-based algorithms for $k$-restricted approximate covers and seeds, respectively.
Experimental evaluation of these algorithms was performed by Guth~\cite{guththesis}.

Recently, Guth~\cite{GUTH2019} extended this study to \emph{$k$-approximate restricted enhanced covers} under Hamming distance.
In this problem, we search for a border of $T$ whose $k$-approximate occurrences cover the maximum number of text positions.
In another variant of the problem, which one could see as approximate partial cover problem, we
only require the approximate enhanced cover to be a $k$-approximate border of $T$, but still to be a factor of $T$.
Guth~\cite{GUTH2019} proposed $\Oh(n^2)$-time and $\Oh(n^3(|\Sigma|+k))$-time algorithms for the two respective variants.

We improve upon previous results on restricted approximate quasiperiodicity.
We introduce a general notion of \emph{$k$-coverage} of a string $S$ in a string $T$, defined as the number of positions in $T$
that are covered by $k$-approximate occurrences of $S$.
Efficient algorithms computing the $k$-coverage for factors of $T$ are presented.
We also show NP-hardness for non-restricted approximate covers and seeds under the Hamming distance.
A detailed list of our results is as follows.

\begin{enumerate}
\item The Hamming $k$-coverage for every prefix and for every factor of a string of length $n$ can be computed in $\Oh(nk^{2/3}\log^{1/3}n\log k)$ time
(for a string over an integer alphabet) and $\Oh(n^2)$ time, respectively.
(See Section~\ref{sec:Hamming}.)

With this result we obtain algorithms with the same time complexities for the two versions of $k$-approximate restricted enhanced covers that were proposed by Guth~\cite{GUTH2019}
and an $\Oh(n^2k)$-time algorithm computing $k$-restricted approximate covers and seeds.
Our algorithm for prefixes actually works in linear time assuming that a $k$-mismatch version of the $\PREF$ table~\cite{Jewels} is given.
Thus, as a by-product, for $k=0$, we obtain an alternative linear-time algorithm for computing all (exact) enhanced covers of a string.
(A different linear-time algorithm for this problem was given in \cite{DBLP:journals/tcs/FlouriIKPPST13}).

The complexities come from using tools of Kaplan et al.~\cite{DBLP:conf/swat/KaplanPS06} and Flouri et al.~\cite{DBLP:journals/ipl/FlouriGKU15}, respectively.

\item The $k$-coverage under Levenshtein distance and weighted edit distance for every factor of a string of length $n$ can be computed in
$\Oh(n^3)$ time and $\Oh(n^3\sqrt{n \log n})$ time, respectively.
(See Section~\ref{sec:edit}.)

We also show in Section~\ref{sec:edit} how our approach can be used to compute restricted approximate covers and seeds under
weighted edit distance in $\Oh(n^3\sqrt{n \log n})$ time, thus improving
upon the previous $\Oh(n^4)$-time algorithms of Sip et al.~\cite{SimParkKimLee} and Christodoulakis et al.~\cite{DBLP:journals/jalc/ChristodoulakisIPS05}.

Our algorithm for Levenshtein distance uses incremental string comparison~\cite{DBLP:journals/siamcomp/LandauMS98}.
\item Under Hamming distance, it is NP-hard to check if a given string of length $n$ has a $k$-approximate cover or a $k$-approximate seed
of a given length $c$.
This statement holds even for strings over a binary alphabet.
(See Section~\ref{sec:NPhard}.)

This result extends the previous proofs of Sip et al.~\cite{SimParkKimLee} and Christodoulakis et al.~\cite{DBLP:journals/jalc/ChristodoulakisIPS05} which worked for the weighted edit distance.
\end{enumerate}

A different notion of approximate cover, which we do not consider in this work,
was recently studied in \cite{DBLP:conf/cpm/AmirLLLP17,DBLP:journals/algorithmica/AmirLLLP19,DBLP:conf/cpm/AmirLLP17,DBLP:journals/tcs/AmirLLP19,DBLP:conf/cpm/AmirLP18}.
This work assumed that the string $T$ may not have a cover, but it is at a small Hamming distance from a string $T'$ that has a proper cover.
They defined an approximate cover of $T$ as the shortest cover of a string $T'$ that is closest to $T$ under Hamming distance.
Interestingly, this problem was also shown to be NP-hard \cite{DBLP:journals/algorithmica/AmirLLLP19} and an $\Oh(n^4)$-time algorithm
was developed for it in the restricted case that the approximate cover is a factor of the string $T$ \cite{DBLP:journals/tcs/AmirLLP19}.
Our work can be viewed as complementary to this study as ``the natural definition of an approximate repetition is not clear'' \cite{DBLP:journals/tcs/AmirLLP19}.

\section{Preliminaries}
We consider strings over an alphabet $\Sigma$.
The empty string is denoted by $\varepsilon$.
For a string $T$, by $|T|$ we denote its length and by $T[0],\ldots,T[|T|-1]$ its subsequent letters.
By $T[i,j]$ we denote the string $T[i]\ldots T[j]$ which we call a \emph{factor} of $T$.
If $i=0$, it is a \emph{prefix} of $T$, and if $j=|T|-1$, it is a \emph{suffix} of $T$.
A string that is both a prefix and a suffix of $T$ is called a \emph{border} of $T$.
For a string $T=XY$ such that $|X|=b$, by $\rot_b(T)$ we denote $YX$, called a \emph{cyclic shift} of $T$.

For equal-length strings $U$ and $V$, by $\Ham(U,V)$ we denote their \emph{Hamming distance}, that is, the number of positions where they do not match.
For strings $U$ and $V$, by $\ed(U,V)$ we denote their \emph{edit distance}, that is,
the minimum cost of edit operations (insertions, deletions, substitutions) that allow to transform $U$ to $V$.
Here the cost of an edit operation can vary depending both on the type of the operation and on the letters that take part in it.
In case that all edit operations have unit cost, the edit distance is also called \emph{Levenshtein distance} and denoted here as $\Lev(U,V)$.

For two strings $S$ and $T$ and metric $d$, we denote by
\[\Occ_k^d(S,T) = \{[i,j]\,:\,d(S,T[i,j]) \le k\}\]
the set of approximate occurrences of $S$ in $T$, represented as intervals, under the metric $d$.
We then denote by 
\[\Covered_k^d(S,T) = |\bigcup \Occ_k^d(S,T)| \]
the \emph{$k$-coverage} of $S$ in $T$.
In case of Hamming or Levenshtein distances, $k \le n$, but for the weighted edit distance $k$ can be arbitrarily large.
Moreover, by $\StartOcc_k^d(S,T)$ we denote the set of left endpoints of the intervals in $\Occ_k^d(S,T)$.

\begin{definition}
Let $d$ be a metric and $T$ be a string.
We say that string $C$, $|C| < |T|$, is a \emph{$k$-approximate cover} of $T$ under metric $d$ if $\Covered_k^d(C,T) = |T|$.
\end{definition}

We say that string $C$, $2|C| \le |T|$, is a \emph{$k$-approximate seed} of $T$ if it is a $k$-approximate cover of some string $T'$ whose factor is $T$.
Let $\diamondsuit$ be a wildcard symbol that matches every other symbol of the alphabet.
Strings over $\Sigma \cup \{\diamondsuit\}$ are also called \emph{partial words}.
In order to compute $k$-approximate seeds, it suffices to consider $k$-approximate covers of $\diamondsuit^{|T|} T \diamondsuit^{|T|}$.

The main problems in scope can now be stated as follows.

\defproblem{General $k$-Approximate Cover/Seed}{
  String $T$ of length $n$, metric $d$, integer $c \in \{1,\ldots,n-1\}$ and number $k$
}{
  A string $C$ of length $c$ that is a $k$-approximate cover/seed of $T$ under $d$
}

\defproblem{Prefix/Factor $k$-Coverage}{
  String $T$ of length $n$, metric $d$ and number $k$
}{
  For every prefix/factor of $T$, compute its $k$-coverage under $d$
}

\defproblem{Restricted Approximate Covers/Seeds}{
  String $T$ of length $n$ and metric $d$
}{
  Compute, for every factor $C$ of $T$, the smallest $k$ such that $C$ is a $k$-approximate cover/seed of $T$ under $d$
}

\subsection{Algorithmic Toolbox for Hamming Distance}
For a string $T$ of length $n$, by $\lcp_k(i,j)$ we denote the length of the longest common prefix with at most $k$ mismatches of the suffixes $T[i,n-1]$ and $T[j,n-1]$.
Flouri et al.~\cite{DBLP:journals/ipl/FlouriGKU15} proposed an $\Oh(n^2)$-time algorithm to compute the longest common factor of two strings $T_1$, $T_2$ with at most $k$ mismatches.
Their algorithm actually computes the lengths of the longest common prefixes with at most $k$ mismatches of every two suffixes $T_1[i,|T_1|-1]$ and $T_2[j,|T_2|-1]$
and returns the maximum among them.
Applied for $T_1=T_2$, it gives the following result.

\begin{lemma}[\cite{DBLP:journals/ipl/FlouriGKU15}]\label{lem:Flouri}
  For a string of length $n$, values $\lcp_k(i,j)$ for all $i,j=0,\ldots,n-1$ can be computed in $\Oh(n^2)$ time.
\end{lemma}

We also use a table $\PREF_k$ such that $\PREF_k[i]=\lcp_k(0, i)$.
LCP-queries with mismatches can be answered in $\Oh(k)$ time after linear-time preprocessing using the kangaroo method~\cite{DBLP:journals/tcs/LandauV86}.
In particular, this allows to compute the $\PREF_k$ table in $\Oh(nk)$ time.
Kaplan et al.~\cite{DBLP:conf/swat/KaplanPS06} presented an algorithm that, given a pattern $P$ of length $m$, a text $T$ of length $n$ over an integer alphabet
$\Sigma \subseteq \{1,\ldots,n^{\Oh(1)}\}$, and an integer $k$,
finds in $\Oh(nk^{2/3}\log^{1/3}m\log k)$ time for all positions $j$ of $T$, the index of the $k$-th mismatch of $P$ with the suffix $T[j,n-1]$.
Applied for $P=T$, it gives the following result.

\begin{lemma}[\cite{DBLP:conf/swat/KaplanPS06}]\label{lem:PREFk}
  The $\PREF_k$ table of a string of length $n$ over an integer alphabet can be computed in $\Oh(nk^{2/3}\log^{1/3}n\log k)$ time.
\end{lemma}

We say that strings $U$ and $V$ have a \emph{$k$-mismatch prefix-suffix of length $p$} if $U$ has a prefix $U'$ of length $p$ and $V$ has a suffix $V'$ of length $p$
such that $\Ham(U',V') \le k$.

\subsection{Algorithmic Toolbox for Edit Distance}
For $x,y \in \Sigma$, let $c(x,y)$, $c(\varepsilon,x)$ and $c(x,\varepsilon)$ be the costs of substituting letter $x$ by letter $y$ (equal to 0 if $x=y$),
inserting letter $x$ and deleting letter $x$, respectively.
They are usually specified by a penalty matrix $c$; it implies a metric if certain conditions are satisfied (identity of indiscernibles, symmetry, triangle inequality).

The classic dynamic programming solution to the edit distance problem (see~\cite{WF}) for strings $T_1$ and $T_2$ uses the so-called $D$-table such that
$D[i,j]$ is the edit distance between prefixes $T_1[0,i]$ and $T_2[0,j]$.
Initially $D[-1,-1]=0$, $D[i,-1]=D[i-1,-1]+c(T_1[i],\varepsilon)$ for $i \ge 0$ and $D[-1,j]=D[-1,j-1]+c(\varepsilon,T_2[j])$ for $j \ge 0$.
For $i,j \ge 0$, $D[i,j]$ can be computed as follows:
\[ D[i,j] = \min(D[i-1,j-1]+c(T_1[i],T_2[j]),\,D[i,j-1]+c(\varepsilon,T_2[j]),\,D[i-1,j]+c(T_1[i],\varepsilon)). \]

Given a threshold $h$ on the Levenshtein distance, Landau et al.~\cite{DBLP:journals/siamcomp/LandauMS98} show how to compute the Levenshtein distance between
$T_1$ and $bT_2$, for any $b \in \Sigma$, in $\Oh(h)$ time using previously
computed solution for $T_1$ and $T_2$ (another solution was given later by Kim and Park~\cite{DBLP:journals/jda/KimP04}).
They define an \emph{$h$-wave} that contains indices of the last value $h$ in diagonals of the $D$-table.
Let $L^h(d) = \max \{ i : D[i, i+d] = h\}$.
Formally an $h$-wave is:
\[L^h = [ L^h(-h), L^h(-h+1),\ldots,L^h(h-1), L^h(h) ].\]
Landau et al.~\cite{DBLP:journals/siamcomp/LandauMS98} show how to update the $h$-wave when string $T_2$ is prepended by a single letter in $\Oh(h)$ time.
This method was introduced to approximate periodicity in~\cite{DBLP:journals/tcs/SimIPS01}.

\section{Computing $k$-Coverage under Hamming Distance}\label{sec:Hamming}
Let $T$ be a string of length $n$ and assume that its $\PREF_k$ table is given.
We will show a linear-time algorithm for computing the $k$-coverage of every prefix of $T$ under the Hamming distance.

In the algorithm we consider all prefix lengths $\ell=1,\ldots,n$.
At each step of the algorithm, a linked list $\L$ is stored that contains all positions $i$ such that $\PREF_k[i] \ge \ell$ and a sentinel value $n$, in an increasing order.
The list is stored together with a table $A(\L)[0..n-1]$ such that $A(\L)[i]$ is a link to the occurrence of $i$ in $\L$ or \textbf{nil} if $i \not\in \L$.
It can be used to access and remove a given element of $\L$ in $\Oh(1)$ time.
Before the start of the algorithm, $\L$ contains all numbers $0,\ldots,n$.

If $i \in \L$ and $j$ is the successor of $i$ in $\L$, then the approximate occurrence of $T[0,\ell-1]$ at position $i$ accounts for $\min(\ell,j-i)$ positions
that are covered in $T$.
A pair of adjacent elements $i<j$ in $\L$ is called \emph{overlapping} if $j-i < \ell$ and \emph{non-overlapping} otherwise.
Hence, each non-overlapping adjacent pair adds the same amount to the number of covered positions.

All pairs of adjacent elements of $\L$ are partitioned in two data structures, $\D_o$ and $\D_{no}$, that store overlapping and non-overlapping pairs, respectively.
Data structure $\D_{no}$ stores non-overlapping pairs $(i,j)$ in buckets that correspond to $j-i$, in a table $B(\D_{no})$ indexed from 1 to $n$.
It also stores a table $A(\D_{no})$ indexed 0 through $n-1$ such that $A(\D_{no})[i]$ points to the location of $(i,j)$ in its bucket, provided that such a pair exists for some $j$,
or \textbf{nil} otherwise.
Finally, it remembers the number $\num(\D_{no})$ of stored adjacent pairs.
$\D_o$ does not store the overlapping adjacent pairs $(i,j)$ explicitly, just the sum of values $j-i$, as $\suma(\D_o)$.
Then
\begin{equation}\label{eq:covered1}
\Covered_k^{\Ham}(T[0,\ell-1],T) = \suma(\D_o) + \num(\D_{no}) \cdot \ell.
\end{equation}

Now we need to describe how the data structures are updated when $\ell$ is incremented.

In the algorithm we store a table $Q[0..n]$ of buckets containing pairs $(\PREF_k[i],i)$ grouped by the first component.
When $\ell$ changes to $\ell+1$, 
the second components of all pairs from $Q[\ell]$ are removed, one by one, from the list $\L$ (using the table $A(\L)$).

Let us describe what happens when element $q$ is removed from $\L$.
Let $q_1$ and $q_2$ be its predecessor and successor in $\L$. (They exist because $0$ and $n$ are never removed from $\L$.)
Then each of the pairs $(q_1,q)$ and $(q,q_2)$ is removed from the respective data structure $\D_o$ or $\D_{no}$, depending on the difference of elements.
Removal of a pair $(i,j)$ from $\D_o$ simply consists in decreasing $\suma(\D_o)$ by $j-i$,
whereas to remove $(i,j)$ from $\D_{no}$ one needs to remove it from the right bucket (using the table $A(\D_{no})$) and decrement $\num(\D_o)$.
In the end, the pair $(q_1,q_2)$ is inserted to $\D_o$ or to $\D_{no}$ depending on $q_2-q_1$.
Insertion to $\D_o$ and to $\D_{no}$ is symmetric to deletion.

When $\ell$ is incremented, non-overlapping pairs $(i,j)$ with $j-i=\ell$ become overlapping.
Thus, all pairs from the bucket $B(\D_{no})[\ell]$ are removed from $\D_{no}$ and inserted to $\D_o$.

This concludes the description of operations on the data structures.
Correctness of the resulting algorithm follows from~\eqref{eq:covered1}.
We analyze its complexity in the following theorem.

\begin{theorem}
  Let $T$ be a string of length $n$.
  Assuming that the $\PREF_k$ table for string $T$ is given, the $k$-coverage of every prefix of $T$ under the Hamming distance can be computed in $\Oh(n)$ time.
\end{theorem}
\begin{proof}
  There are up to $n$ removals from $\L$.
  Initially $\L$ contains $n$ adjacent pairs.
  Every removal from $\L$ introduces one new adjacent pair, so the total number of adjacent pairs that are considered in the algorithm is $2n-1$.
  Each adjacent pair is inserted to $\D_o$ or to $\D_{no}$, then it may be moved from $\D_{no}$ to $\D_o$, and finally it is removed from its data structure.
  In total, $\Oh(n)$ insertions and deletions are performed on the two data structures, in $\Oh(1)$ time each.
  This yields the desired time complexity of the algorithm.
\end{proof}

Let us note that in order to compute the $k$-coverage of all factors of $T$ that start at a given position $i$, it suffices to
use a table $[\lcp_k(i,0),\ldots,\lcp_k(i,n-1)]$ instead of $\PREF_k$.
Together with Lemma~\ref{lem:Flouri} this gives the following result.

\begin{corollary}
  Let $T$ be a string of length $n$.
  The $k$-coverage of every factor of $T$ under the Hamming distance can be computed in $\Oh(n^2)$ time.
\end{corollary}

\section{Computing $k$-Coverage under Edit Distance}\label{sec:edit}

Let us state an abstract problem that, to some extent, is a generalization of the $k$-mismatch $\lcp$-queries to the edit distance.

\defproblem{Longest Approximate Prefix Problem}{
  A string $T$ of length $n$, a metric $d$ and a number $k$
}{
  A table $P_k^d$ such that $P_k^d[a,b,a']$ is the maximum $b'\ge a'-1$ such that $d(T[a,b],T[a',b']) \le k$ or $-1$ if no such $b'$ exists.
}

Having the table $P_k^d$, one can easily compute the $k$-coverage of a factor $T[a,b]$ under metric $d$ as:
\begin{equation}\label{eq:CovEd}
\Covered_k^d(T[a,b],T)\,=\,\left|\bigcup_{a'=0}^{n-1} [a',P_k^d[a,b,a']]\right|,
\end{equation}
where an interval of the form $[a',b']$ for $b'<a'$ is considered to be empty.
The size of the union of $n$ intervals can be computed in $\Oh(n)$ time, which gives $\Oh(n^3)$ time over all factors.

In Section~\ref{subsec:Lev} and~\ref{subsec:ed} we show how to compute the tables $P_k^{\Lev}$ and $P_k^{\ed}$ for a given threshold $k$
in $\Oh(n^3)$ and $\Oh(n^3\sqrt{n\log n})$ time, respectively.
Then in Section~\ref{subsec:Costas} we apply the techniques of Section~\ref{subsec:ed} to obtain an $\Oh(n^3\sqrt{n\log n})$-time algorithm
for computing restricted approximate covers and seeds under the edit distance.

\subsection{Longest Approximate Prefix under Levenshtein Distance}\label{subsec:Lev}

Let $H_{i,j}$ be the $h$-wave for strings $T[i, n-1]$ and $T[j, n-1]$ and $h=k$. Then we can compute $P_k^{\Lev}$ with Algorithm~\ref{algo:Lev}.
The algorithm basically takes the rightmost diagonal of $D$-table in which the value in row $b-a+1$ is less than or equal to $k$.

\begin{algorithm}[h!]
\For{$a':=n-1$ \KwSty{down to} $0$}{
	Compute $H_{n-1,a'}$\;
	\For{$a:=n-1$ \KwSty{down to} $0$}{
		\If{$a < n$} {Compute $H_{a,a'}$ from $H_{a+1,a'}$\;}

		$d:=k$\;
		\For {$b:=a$ \KwSty{to} $n-1$}{
			$i := b - a + 1$\; 
			\While {$d \geq -k$ \KwSty{and} {$H_{a,a'}(d) < i$} }
			{ $d:=d-1$\;}

			\lIf {$d < -k$} {$P_k^{\Lev}[a, b, a'] := -1$}
			\lElse {$P_k^{\Lev}[a, b, a'] := a' + i + d$}
		}
	}
}
\caption{Computing $P_k^{\Lev}$ table.}
\label{algo:Lev}
\end{algorithm}	

The while-loop can run up to $2k$ times for given $a$ and $a'$.
Computing $H_{n-1,a'}$ takes $\Oh(k^2)$ time and updating $H_{a,a'}$ takes $\Oh(k)$ time. It makes the algorithm run in $\Oh(n^3)$ time.
Together with Equation~\eqref{eq:CovEd} this yields the following result.

\begin{proposition}
  Let $T$ be a string of length $n$.
  The $k$-coverage of every factor of $T$ under the Levenshtein distance can be computed in $\Oh(n^3)$ time.
\end{proposition}

A similar method could be used in case of constant edit operation costs, by applying the work of~\cite{DBLP:journals/jda/HyyroNI15}.
In the following section we develop a solution for arbitrary costs.

\subsection{Longest Approximate Prefix under Edit Distance}\label{subsec:ed}

For indices $a,a' \in [0,n]$ we define a table $D_{a,a'}$ such that $D_{a,a'}[b,b']$ is the edit distance between $T[a,b]$ and $T[a',b']$, for $b \in [a-1,n-1]$ and $b' \in [a'-1,n-1]$.
For other indices we set $D_{a,a'}[b,b']=\infty$.
The $D_{a,a'}$ table corresponds to the $D$-table for $T[a,n-1]$ and $T[a',n-1]$ and so it can be computed in $\Oh(n^2)$ time.

We say that pair $(d,b)$ (Pareto-)dominates pair $(d',b')$ if $(d,b) \ne (d',b')$, $d \le d'$ and $b \ge b'$.
Let us introduce a data structure $L_{a,a'}[b]$ being a table of all among pairs $(D_{a,a'}[b,b'],b')$ that are maximal in this sense (i.e., are not dominated by other pairs),
sorted by increasing first component.
Using a folklore stack-based algorithm (Algorithm~\ref{algo2}), this data structure can be computed from $D_{a,a'}[b,a'-1],\ldots,D_{a,a'}[b,n-1]$ in linear time.

\begin{algorithm}[h!]
$Q:=\,$empty stack\;
\For{$b':=a'-1$ \KwSty{to} $n-1$}{
  $d:=D_{a,a'}[b,b']$\;
  \While{$Q$ not empty}{
    $(d',x):=\mathit{top}(Q)$\;
    \lIf{$d'\ge d$}{$\mathit{pop}(Q)$}\lElse{\KwSty{break}}
  }
  $\mathit{push}(Q,(d,b'))$\;
}
$L_{a,a'}[b]:=Q$\;
\caption{Computing $L_{a,a'}[b]$ from $D_{a,a'}[b,\cdot]$.}\label{algo2}
\end{algorithm}

Every multiple of $M=\floor{\sqrt{n / \log n}}$ will be called a \emph{special point}.
In our algorithm we first compute the following data structures:
\begin{enumerate}[(a)]
\item\label{it:a} all $L_{a,a'}[b]$ lists where $a$ or $a'$ is a special point, for $a,a' \in [0,n-1]$ and $b \in [a-1,n-1]$ (if $a\ge n$ or $a'\ge n$, the list is empty); and
\item\label{it:b} all cells $D_{a,a'}[b,b']$ of all $D_{a,a'}$ tables for $a,a' \in [0,n]$ and $-1 \le b-a,b'-a' < M-1$.
\end{enumerate}
Computing part \eqref{it:a} takes $\Oh(n^4/M) = \Oh(n^3\sqrt{n\log n})$ time, whereas part \eqref{it:b} can be computed in $\Oh(n^4/M^2) = \Oh(n^3\log n)$ time.
The intuition behind this data structure is shown in the following lemma.

\begin{lemma}\label{lem:special}
  Assume that $b-a \ge M-1$ or $b'-a' \ge M-1$.
  Then there exists a pair of positions $c$, $c'$ such that the following conditions hold:
  \begin{itemize}
  \item $a \le c \le b+1$ and $a' \le c' \le b'+1$, and
  \item $c-a,c'-a' < M$, and
  \item $\ed(T[a,b],T[a',b']) = \ed(T[a,c-1],T[a',c'-1]) + \ed(T[c,b],T[c',b'])$, and
  \item at least one of $c$, $c'$ is a special point.
  \end{itemize}
  Moreover, if $c$ ($c'$) is the special point, then $c \le b$ ($c'\le b'$, respectively).
\end{lemma}
\begin{proof}
  By the assumption, at least one of the intervals $[a,b]$ and $[a',b']$ contains a special point.
  Let $p \in [a,b]$ and $p' \in [a',b']$ be the smallest among them; we have $p-a,p'-a' < M$ provided that $p$ or $p'$ exists, respectively (otherwise $p$ or $p'$ is set to $\infty$).
  Let us consider the table $D_{a,a'}$ and how its cell $D_{a,a'}[b,b']$ is computed.
  We can trace the path of parents in the dynamic programming from $D_{a,a'}[b,b']$ to the origin ($D_{a,a'}[a-1,a'-1]$).
  Let us traverse this path in the reverse direction until the first dimension of the table reaches $p$ or the second dimension reaches $p'$.
  Say that just before this step we are at $D_{a,a'}[q,q']$.
  If $q+1=p$ and $q'<p'$, then we set $c=q+1$ and $c'=q'+1$.
  Indeed $c=p$ is a special point,
  \[ \ed(T[a,b],T[a',b']) = \ed(T[a,c-1],T[a',c'-1]) + \ed(T[c,b],T[c',b']) \]
  and $c-a,c'-a' < M$.
  Moreover, $q' \in [a'-1,b']$, so $c' \in [a',b'+1]$.
  The opposite case (that $q'+1=p'$) is symmetric.
\end{proof}

If $P_k^{\ed}[a,b,a']-a' < M-1$, then it can be computed using one the $M \times M$
prefix fragments of the $D_{a,a'}$ tables.
Otherwise, according to the statement of the lemma, one of the $L_{c,c'}[b]$ lists can be used, where $c-a,c'-a' < M$, as shown 
in Algorithm~\ref{alg:weighted}.
The algorithm uses a predecessor operation $\Pred(x,L)$ which for a number $x$ and a list $L=L_{c,c'}[b]$ returns the maximal pair whose first component does not exceed $x$,
or $(\infty,\infty)$ if no such pair exists.
This operation can be implemented in $\Oh(\log n)$ time via binary search.

\begin{algorithm}[ht!]
$\res:=-1$\;
\If{$b-a<M-1$}{ \label{ll2}
  \For{$b':=a'-1$ \KwSty{to} $a'+M-2$}{
    \If{$D_{a,a'}[b,b'] \le k$}{\label{ll4}
      $\res:=b'$\; \label{ll5}
    }
  }
}
$s:=a+((-a)\bmod M)$; $s':=a'+((-a')\bmod M)$\tcp*{closest special points}\label{ll6}
\ForEach{$(c,c')$ \KwSty{in} $(\{s\} \times [a',a'+M-1]) \cup ([a,a+M-1] \times \{s'\})$\label{ll7}}{
  $(d',b'):=\Pred(k - D_{a,a'}[c-1,c'-1],\,L_{c,c'}[b])$\;\label{ll8}
  \If{$d' \ne \infty$}{
    $\res:=\max(\res,b')$\;\label{ll10}
  }
}
$P_k^{\ed}[a,b,a']:=\res$\;
\caption{Computing $P_k^{\ed}[a,b,a']$.}
\label{alg:weighted}
\end{algorithm}

\begin{theorem}\label{thm:3}
  Let $T$ be a string of length $n$.
  The $k$-coverage of every factor of $T$ under the edit distance can be computed in $\Oh(n^3\sqrt{n\log n})$ time.
\end{theorem}
\begin{proof}
  We want to show that Algorithm~\ref{alg:weighted} correctly computes $P_k^{\ed}[a,b,a']$.
  Let us first check that the result $b'=\res$ of Algorithm~\ref{alg:weighted} satisfies $D_{a,a'}[b,b'] \le k$.
  It is clear if the algorithm computes $b'$ in line~\ref{ll5}.
  Otherwise, it is computed in line~\ref{ll10}.
  This means that $L_{c,c'}[b]$ contains a pair $(D_{c,c'}[b,b'],b')$ such that
  \[ k\ \ge\ D_{c,c'}[b,b'] + D_{a,a'}[c-1,c'-1]\ \ge\ D_{a,a'}[b,b']. \]

  Now we show that the returned value $\res$ is at least $x=P_k^{\ed}[a,b,a']$.
  If $b-a < M-1$ and $x-a'<M-1$, then the condition in line~\ref{ll4} holds for $b'=x$, so indeed $\res \ge x$.
  Otherwise, the condition of Lemma~\ref{lem:special} is satisfied.
  The lemma implies two positions $c,c'$ such that at least one of them is special and that satisfy additional constraints.

  If $c$ is special, then the constraints $a \le c$ and $c-a<M$ imply that $c=s$, as defined in line~\ref{ll6}.
  Additionally, $a' \le c' \le a'+M-1$, so $(c,c')$ will be considered in the loop from line~\ref{ll7}.
  By the lemma and the definition of $x$, we have
  \begin{equation}\label{eq:dddd} D_{c,c'}[b,x]\ =\ D_{a,a'}[b,x] - D_{a,a'}[c-1,c'-1]\ \le\ k - D_{a,a'}[c-1,c'-1].\end{equation}
  The list $L_{c,c'}[b]$ either contains the pair $(D_{c,c'}[b,x],x)$,
  or a pair $(D_{c,c'}[b,x'],x')$ such that $D_{c,c'}[b,x'] \le D_{c,c'}[b,x]$ and $x'>x$.
  In the latter case by~\eqref{eq:dddd} we would have
  \[k\ \ge\ D_{a,a'}[c-1,c'-1] + D_{c,c'}[b,x]\ \ge\ D_{a,a'}[c-1,c'-1] + D_{c,c'}[b,x']\ \ge\ D_{a,a'}[b,x'] \]
  and $x'>x$.
  In both cases the predecessor computed in line~\ref{ll8} returns a value $\res$ such that $\res \ge x$ and $\res \ne \infty$.
  The case that $c'$ is special admits an analogous argument.

  Combining Algorithm~\ref{alg:weighted} with Equation~\eqref{eq:CovEd}, we obtain correctness of the computation.

  As for complexity,
  Algorithm~\ref{alg:weighted} computes $P_k^{\ed}[a,b,a']$ in $\Oh(M \log n) = \Oh(\sqrt{n\log n})$ time and
  the pre-computations take $\Oh(n^3\sqrt{n\log n})$ total time.
\end{proof}

\subsection{Restricted Approximate Covers and Seeds under Edit Distance}\label{subsec:Costas}
The techniques that were developed in Section~\ref{subsec:ed} can be used to improve upon the $\Oh(n^4)$ time complexity
of the algorithms for computing the restricted approximate covers and seeds under the edit distance~\cite{DBLP:journals/jalc/ChristodoulakisIPS05,SimParkKimLee}.
We describe our solution only for restricted approximate covers; the solution for restricted approximate seeds follows by considering the text $\diamondsuit^{|T|} T \diamondsuit^{|T|}$.

Let us first note that the techniques from the previous subsection can be used as a black box to solve the problem in scope in $\Oh(n^3\sqrt{n\log n} \log (nw))$ time,
where $w$ is the maximum cost of an edit operation.
Indeed, for every factor $T[a,b]$ we binary search for the smallest $k$ for which $T[a,b]$ is a $k$-approximate cover of $T$.
A given value $k$ is tested by computing the tables $P_k^{\ed}[a,b,a']$ for all $a'=0,\ldots,n-1$ and checking if $\Covered_k^d(T[a,b],T)=n$ using Equation~\eqref{eq:CovEd}.

Now we proceed to a more efficient solution.
Same as in the algorithms from~\cite{DBLP:journals/jalc/ChristodoulakisIPS05,SimParkKimLee} we compute,
for every factor $T[a,b]$, a table $Q_{a,b}[0..n]$ such that $Q_{a,b}[i]$ is the minimum edit distance threshold $k$
for which $T[a,b]$ is a $k$-approximate cover of $T[i,n-1]$.
In the end, all factors $T[a,b]$ for which $Q_{a,b}[0]$ is minimal need to be reported as restricted approximate covers of $T$.
We will show how, given the data structures \eqref{it:a} and \eqref{it:b} of the previous section, we can compute this table in $\Oh(n M \log n)$ time.

A dynamic programming algorithm for computing the $Q_{a,b}$ table,
similar to the one in~\cite{DBLP:journals/jalc/ChristodoulakisIPS05}, is shown in Algorithm~\ref{alg:n2}.
Computing $Q_{a,b}$ takes $\Oh(n^2)$ time provided that all $D_{a,b}$ arrays, of total size $\Oh(n^4)$, are available.
The algorithm considers all possibilities for the approximate occurrence $T[i,j]$ of $T[a,b]$.

\newcommand{\minQ}{\mathit{minQ}}

\begin{algorithm}[ht!]
$Q_{a,b}[n]:=0$\;
\For{$i:=n-1$ \KwSty{down to} $0$}{
  $Q_{a,b}[i]:=\infty$\;
  $\minQ := \infty$\;
  \For{$j:=i$ \KwSty{to} $n-1$}{
    $\minQ := \min(\minQ,Q_{a,b}[j+1])$\tcp*{$\minQ = \min Q_{a,b}[i+1..j+1]$}
    $Q_{a,b}[i]:=\min(Q_{a,b}[i],\,\max(D_{a,i}[b,j],\minQ))$\; \label{l7}
  }
}
\caption{Computing $Q_{a,b}$ in quadratic time.}
\label{alg:n2}
\end{algorithm}

During the computation of $Q_{a,b}$, we will compute a data structure for on-line range-minimum queries over the table.
We can use the following simple data structure with $\Oh(n \log n)$ total construction time and $\Oh(1)$-time queries.
For every position $i$ and power of two $2^p$, we store as $\RM[i,p]$ the minimal value in the table $Q_{a,b}$ on the interval $[i,i+2^p-1]$. 
When a new value $Q_{a,b}[i]$ is computed, we compute
$\RM[i,0]=Q_{a,b}[i]$ and
$\RM[i,p]$ for all $0 < p \le \log_2 (n-i)$ using the formula $\RM[i,p] = \min(\RM[i,p-1],\RM[i+2^{p-1},p-1])$.
Then a range-minimum query over an interval $[i,j]$ of $Q_{a,b}$ can be answered by inspecting up to two cells of the $\RM$ table for $p$ such that
$2^p \le j-i+1 < 2^{p+1}$.

Let us note that the variable $\minQ$, which denotes the minimum of a growing segment in the $Q_{a,b}$ table, can only decrease.
We would like to make the second argument of $\max$ in line~\ref{l7} non-decreasing for increasing $j$.
The values $\ed(T[a,b],T[i,j]) = D_{a,i}[b,j]$ may increase or decrease as $j$ grows.
However, it is sufficient to consider only those values of $j$ for which $(D_{a,i}[b,j],j)$ is not (Pareto-)dominated (as in Section~\ref{subsec:ed}),
i.e., the elements of the list $L_{a,i}[b]$.
For these values, $D_{a,i}[b,j]$ is indeed increasing for increasing $j$.
The next observation follows from this monotonicity and the monotonicity of $\min Q_{a,b}[i+1..j+1]$.

\begin{observation}
Let $(D_{a,i}[b,j'],j')$ be the first element on the list $L_{a,i}[b]$ such that
\[\min Q_{a,b}[i+1..j'+1] \le D_{a,i}[b,j'].\]
If $j'$ does not exist, we simply take the last element of $L_{a,i}[b]$.
Further let $(D_{a,i}[b,j''],j'')$ be the predecessor of $(D_{a,i}[b,j'],j')$ in $L_{a,i}[b]$ (if it exists).
Then $j \in \{j',j''\}$ minimizes the value of the expression $\max(\min Q_{a,b}[i+1..j+1], D_{a,i}[b,j])$.
\end{observation}

If we had access to the list $L_{a,i}[b]$, we could use binary search to locate the index $j'$ defined in the observation.
However, we only store the lists $L_{a,i}[b]$ for $a$ and $i$ such that at least one of them is a special point.
We can cope with this issue by separately considering all $j$ such that $j < i + M-1$ and then performing binary search on every
of $\Oh(M)$ lists $L_{c,c'}[b]$ where $a \le c < a+M$, $i \le c' < i+M$ and at least one of $c$, $c'$ is a special point,
just as in Algorithm~\ref{alg:weighted}.
A pseudocode of the resulting algorithm is given as Algorithm~\ref{alg:fullps}.

\begin{algorithm}[ht!]
$Q_{a,b}[n]:=0$\;
\For{$i:=n-1$ \KwSty{down to} $0$}{
  $Q_{a,b}[i]:=\infty$\;
  $\minQ := \infty$\;
  \If{$b-a<M-1$}{
    \For{$j:=i$ \KwSty{to} $i+M-2$}{
      $\minQ := \min(\minQ,Q_{a,b}[j+1])$\;
      $Q_{a,b}[i]:=\min(Q_{a,b}[i],\,\max(D_{a,i}[b,j],\minQ))$\;
    }
  }
  $s:=a+((-a)\bmod M)$; $s':=i+((-i)\bmod M)$\;
  \ForEach{$(c,c')$ \KwSty{in} $(\{s\} \times [i,i+M-1]) \cup ([a,a+M-1] \times \{s'\})$}{
    \lIf{$L_{c,c'}[b]$ is empty}{\KwSty{continue}}
    \tcc{Binary search}
    $(d_{j'},j'):=\,$ the first pair in $L_{c,c'}[b]$ such that $\min Q_{a,b}[i+1..j'+1] \le D_{a,i}[c-1,c'-1]+d_{j'}$ or the last pair\;
    $(d_{j''},j''):=\,$predecessor of $(d_{j'},j')$ in $L_{c,c'}[b]$ or $(d_{j'},j')$ if there is none\;
    \ForEach{$j$ \KwSty{in} $\{j',j''\}$}{
      $Q_{a,b}[i]:=\min(Q_{a,b}[i],\,\max(D_{a,i}[c-1,c'-1]+d_j, \min Q_{a,b}[i+1..j+1]))$\;
    }
  }
}
\caption{Computing $Q_{a,b}$ in $\Oh(n\sqrt{n \log n})$ time using pre-computed data structures.}
\label{alg:fullps}
\end{algorithm}

Let us summarize the complexity of the algorithm.
Pre-computation of auxiliary data structures requires $\Oh(n^3\sqrt{n\log n})$ time.
Then for every factor $T[a,b]$ we compute the table $Q_{a,b}$.
The data structure for constant-time range-minimum queries over the table costs only additional $\Oh(n \log n)$ space and computation time.
When computing $Q_{a,b}[i]$ using dynamic programming, we may separately consider first $M-1$ indices $j$,
and then we perform a binary search in $\Oh(M)$ lists $L_{c,c'}[b]$.
In total, the time to compute $Q_{a,b}[i]$ given $a$, $b$, $i$ is $\Oh(M \log n) = \Oh(\sqrt{n \log n})$.

\begin{theorem}
  Let $T$ be a string of length $n$.
  All restricted approximate covers and seeds of~$T$ under the edit distance can be computed in $\Oh(n^3\sqrt{n\log n})$ time.
\end{theorem}

The work of~\cite{DBLP:journals/jalc/ChristodoulakisIPS05,SimParkKimLee} on approximate covers and seeds originates from a study of approximate periods~\cite{DBLP:journals/tcs/SimIPS01}.
Interestingly, while our algorithm improves upon the algorithms for computing approximate covers and seeds, it does not work for approximate periods.

\section{NP-hardness of General Hamming $k$-Approximate Cover and Seed}\label{sec:NPhard}

We make a reduction from the following problem.

\defproblem{Hamming String Consensus}{
  Strings $S_1,\ldots,S_m$, each of length $\ell$, and an integer $k \le \ell$
}{
  A string $S$, called consensus string, such that $\Ham(S,S_i) \le k$ for all $i=1,\ldots,m$
}

The following fact is known.

\begin{fact}[\cite{DBLP:journals/mst/FrancesL97}]\label{fct:Frances}
  \textsc{Hamming String Consensus} is NP-complete even for the binary alphabet.
\end{fact}

Let strings $S_1,\ldots,S_m$ of length $\ell$ over the alphabet $\Sigma=\{0,1\}$ and integer $k$ be an instance of \textsc{Hamming String Consensus}.
We introduce a morphism $\phi$ such that
\[\phi(0)=0^{2k+4}\,1010\,0^{2k+4},\quad \phi(1)=0^{2k+4}\,1011\,0^{2k+4}.\]
We will exploit the following simple property of this morphism.
\begin{observation}\label{obs:phi}
  For every string $S$, every length-$(2k+4)$ factor of $\phi(S)$ contains at most three ones.
\end{observation}

Let
$\gamma_i = 1^{2k+4}\phi(S_i)$
and let $\psi(U)$ be an operation that reverses this encoding, i.e., $\psi(\gamma_i) = S_i$.
Formally, it takes as input a string $U$ and outputs $U[4k+12-1]U[2\cdot(4k+12)-1]\ldots U[(\ell-1)(4k+12)-1]$.

\begin{lemma}\label{lem:NP_cover}
  Strings $\gamma_i$ and $\gamma_j$, for any $i,j \in \{1,\ldots,m\}$, have no $2k$-mismatch prefix-suffix of length $p \in \{2k+4,\ldots,|\gamma_i|-1\}$.
\end{lemma}
\begin{proof}
  We will show that the prefix $U$ of $\gamma_i$ of length $p$ and the suffix $V$ of $\gamma_j$ of length $p$ have at least $2k+1$ mismatches.
  Let us note that $U$ starts with $1^{2k+4}$.
  The proof depends on the value $d=|\gamma_i|-p$; we have $1 \le d \le |\gamma_i|-2k-4$.
  Let us start with the following observation that can be readily verified.
  \begin{observation}
    For $A,B \in \{1010,1011\}$, the strings $A0^4$ and $0^4B$ have no 1-mismatch prefix-suffix of length in $\{5,\ldots,8\}$.
  \end{observation}
  
  If $1 \le d \le 4$, then $U$ and $V$ have a mismatch at position $2k+4$ since $V$ starts with $1^{2k+4-d}0$.
  Moreover, they have at least $2\ell$ mismatches by the observation (applied for the prefix-suffix length $d+4$).
  In total, $\Ham(U,V) \ge 2\ell+1 \ge 2k+1$.

  If $4 < d < 2k+4$, then every block $1010$ or $1011$ in $\gamma_i$ and in $\gamma_j$ is matched against a block of zeroes in the other string, which gives
  at least $4\ell$ mismatches.
  Hence, $\Ham(U,V) \ge 4\ell \ge 2k+1$.

  Finally, if $2k+4 \le d \le |\gamma_i|-2k-4$, then $U$ starts with $1^{2k+4}$ and every factor of $V$ of length $2k+4$ has at most three ones (see Observation~\ref{obs:phi}).
  Hence, $\Ham(U,V) \ge 2k+1$.
\end{proof}

We set $T=\gamma_1 \ldots \gamma_m$.
The following lemma gives the reduction.

\begin{lemma}\label{lem:Hamming_cover_NP}
  If \textsc{Hamming String Consensus} for $S_1,\ldots,S_m$, $\ell$, $k$ has a positive answer, then the \textsc{General $k$-Approximate Cover} under Hamming distance for $T$, $k$, and $c=|\gamma_i|$
  returns a $k$-approximate cover $C$ such that $S=\psi(C)$ is a Hamming consensus string for $S_1,\ldots,S_m$.
\end{lemma}
\begin{proof}
  By Lemma~\ref{lem:NP_cover}, 
  if $C$ is a $k$-approximate cover of $T$ of length $c$, then
  every position $a \in \StartOcc_k^H(C,T)$ satisfies $c \mid a$.
  Hence,
  $\StartOcc_k^H(C,T) = \{0,c,2c,\ldots,(m-1)c\}$.
  
  If \textsc{Hamming String Consensus} for $S_1,\ldots,S_m$ has a positive answer $S$, then $1^{2k+4}\phi(S)$ is a $k$-approximate cover of $T$ of length $c$.
  Moreover, if $T$ has a $k$-approximate cover $C$ of length $c$, then for $S = \psi(C)$ and for each $i=1,\ldots,m$, we have that
  \[\Ham(C,T[(i-1)c,ic-1]) \ge \Ham(S,S_i),\]
  so $S$ is a consensus string for $S_1,\ldots,S_m$.
  This completes the proof.
\end{proof}

Lemma~\ref{lem:Hamming_cover_NP} and Fact~\ref{fct:Frances} imply that computing $k$-approximate covers is NP-hard.
Obviously, it is in NP.

\begin{theorem}
  \textsc{General $k$-Approximate Cover} under the Hamming distance is NP-complete even over a binary alphabet.
\end{theorem}

A lemma that is similar to Lemma~\ref{lem:Hamming_cover_NP} can be shown for approximate seeds.
Let
$$T' = \gamma_1\gamma_1 \ldots \gamma_m 1^{2k+4} \gamma_m 1^{2k+4}.$$

\begin{lemma}\label{lem:Hamming_seed_NP}
  If \textsc{Hamming String Consensus} for $S_1,\ldots,S_m$, $\ell$, $k$ has a positive answer, then the \textsc{General $k$-Approximate Seed} under Hamming distance for $T'$, $k$, and $c=|\gamma_1|+2k+4$
  returns a $k$-approximate seed $C$ such that $S=\psi(C')$ is a Hamming consensus string for $S_1,\ldots,S_m$, for some cyclic shift $C'$ of $C$.
\end{lemma}
\begin{proof}
  Assume that $C$ is a $k$-approximate seed of $T'$ of length $c$ and let us consider the approximate occurrence of $C$ that covers position $c-1$ in $T'$.
  Note that it has to be a full occurrence.
  It follows from the next claim that the position of this occurrence is in $\{0,\ldots,2k+3\} \cup \{|\gamma_1|-2k-4,\ldots,|\gamma_1|+2k-3\}$.

  \begin{claim}
    Let $X$ be any length-$c$ factor of $\phi(S_1) 1^{2k+4} \phi(S_1)$
    and $Y$ be any length-$c$ factor of $(\gamma_m 1^{2k+4})^2$.
    Then $\Ham(X,Y) > 2k$.
  \end{claim}
  \begin{proof}
    Let us note that the string $(\gamma_m 1^{2k+4})^2$ contains a middle block $1^{4k+8}$.
    If $Y$ contains this whole block, then certainly $\Ham(X,Y) \ge 2k+1$, since every factor of $X$ of length $4k+8$ contains at most $2k+7$ ones (see Observation~\ref{obs:phi}).
    Otherwise,
    \[
      Y\,=\,1^{2k+4+b} \phi(S_m) 1^{2k+4-b} \quad\text{or}\quad Y\,=\,1^{2k+4-b} \phi(S_m) 1^{2k+4+b}
    \]
    for some $b \in \{0,\ldots,2k+3\}$.
    In particular, $Y$ has $1^{2k+4}$ as a prefix or as a suffix.
    By comparing lengths we see that the length-$(2k+4)$ prefix and suffix of $X$ are factors of $\phi(S_1)$.
    Hence, each of them contains at most three ones (see Observation~\ref{obs:phi}) and $\Ham(X,Y) \ge 2k+1$.
  \end{proof}

  We have established that $C$ has to match, up to at most $k$ mismatches, a string of the form
  \[
    1^b \phi(S_1) 1^{2k+4} 0^{2k+4-b} \quad\mbox{or}\quad 0^b 1^{2k+4} \phi(S_1) 1^{2k+4-b}
  \]
  for some $b \in \{0,\ldots,2k+4\}$.
  We consider the second case; a proof for the first case is analogous (using strings $\gamma'_i = \phi(S_i) 1^{2k+4}$ instead of $\gamma_i$).

  In the second case, $\Ham(C,0^b \gamma_1 1^{2k+4-b}) \le k$.
  Applying Lemma~\ref{lem:NP_cover} for $\gamma_1$ and every $\gamma_j$,
  we get that the starting position $p$ of an occurrence of $C$ in $T'$ that covers the first zero of $\gamma_j$ in the factor $\gamma_1 \ldots \gamma_m$ of $T'$
  has to satisfy $p \equiv -b \bmod {|\gamma_1|}$.

  If \textsc{Hamming String Consensus} for $S_1,\ldots,S_m$ has a positive answer $S$, then the string $1^{2k+4}\phi(S)1^{2k+4}$
  is a $k$-approximate cover (hence, $k$-approximate seed) of $T'$ of length $c$.
  Moreover, if $T'$ has a $k$-approximate seed $C$ of length $c$ such that $\Ham(C,0^b \gamma_1 1^{2k+4-b}) \le k$, then
  for a cyclic shift $C'=\rot_b(C)$, $S = \psi(C')$ and for each $i=1,\ldots,m$, we have that
  \[\Ham(C,\,T[i|\gamma_1|-b,(i+1)|\gamma_1|+2k+4-b]) \ge \Ham(S,S_i),\]
  so $S$ is a consensus string for $S_1,\ldots,S_m$.
  This completes the proof.
\end{proof}
\begin{theorem}
  \textsc{General $k$-Approximate Seed} under the Hamming distance is NP-complete even over a binary alphabet.
\end{theorem}

\section{Conclusions}
We have presented several polynomial-time algorithms for computing restricted approximate covers and seeds and $k$-coverage under Hamming, Levenshtein and weighted edit distances
and shown NP-hardness of non-restricted variants of these problems under the Hamming distance.
It is not clear if any of the algorithms are optimal.
The only known related conditional lower bound shows hardness of computing the Levenshtein distance of two strings in strongly subquadratic time~\cite{DBLP:journals/siamcomp/BackursI18};
however, our algorithms for approximate covers under edit distance work in $\Omega(n^3)$ time.
An interesting open problem is if restricted approximate covers or seeds under Hamming distance, as defined in \cite{DBLP:journals/jalc/ChristodoulakisIPS05,SimParkKimLee},
can be computed in $\Oh(n^{3-\epsilon})$ time, for any $\epsilon>0$.
Here we have shown an efficient solution for $k$-restricted versions of these problems.

\bibliographystyle{plainurl}
\bibliography{approx_quasi_arxiv}

\end{document}